%% file: main_conf.tex
\pdfoutput=1
\documentclass{article}
\usepackage{spconf,amsmath}
\usepackage[dvipdfmx]{graphicx}
\usepackage{subfigure}
\usepackage{amssymb}
\usepackage{amsthm}
\usepackage{color}
\usepackage{cite}
\usepackage[ruled,vlined,linesnumbered]{algorithm2e}

\input{defs}

\newtheorem{proposition}{Proposition}
\newtheorem{lemma}{Lemma}

\title{Efficient Directed Graph Sampling via Gershgorin Disc Alignment}
\name{
    Yuejiang Li$^\star$, 
    H. Vicky Zhao$^\star$, 
    Gene Cheung$^\dag$\thanks{Gene Cheung acknowledges the support of the NSERC grants RGPIN-2019-06271,  RGPAS-2019-00110.}
}
\address{
  $^\star$Dept. of Automation, Tsinghua University, Beijing, China\\ 
  $^\dag$York University, Toronto, Canada
}
\begin{document}
\ninept
\maketitle
\begin{abstract}
Graph sampling is the problem of choosing a node subset via sampling matrix $\H \in \{0,1\}^{K \times N}$ to collect samples $\y = \H \x \in \mathbb{R}^K$, $K < N$, so that the target signal $\x \in \mathbb{R}^N$ can be reconstructed in high fidelity. 
While sampling on undirected graphs is well studied, we propose the first sampling scheme tailored specifically for directed graphs, leveraging a previous undirected graph sampling method based on Gershgorin disc alignment (GDAS). 
Concretely, given a directed positive graph $\cG^d$ specified by random-walk graph Laplacian matrix $\L_{rw}$, we first define reconstruction of a smooth signal $\x^*$ from samples $\y$ using graph shift variation (GSV) $\|\L_{rw} \x\|^2_2$ as a signal prior.
To minimize worst-case reconstruction error of the linear system solution $\x^* = \C^{-1} \H^\top \y$ with symmetric coefficient matrix $\C = \H^\top \H + \mu \L_{rw}^\top \L_{rw}$, the sampling objective is to choose $\H$ to maximize the smallest eigenvalue $\lambda_{\min}(\C)$ of $\C$.
To circumvent eigen-decomposition entirely, we maximize instead a lower bound $\lambda^-_{\min}(\S\C\S^{-1})$ of $\lambda_{\min}(\C)$---smallest Gershgorin disc left-end of a similarity transform of $\C$---via a variant of GDAS based on Gershgorin circle theorem (GCT).
Experimental results show that our sampling method yields smaller signal reconstruction errors at a faster speed compared to competing schemes.
\end{abstract}
\section{Introduction}
\label{sec:intro}
\input{intro}

\section{Preliminaries}
\label{sec:prelim}
\input{prelim}

\section{Problem Formulation}
\label{sec:problem}
\input{problem}

\section{The Proposed GDA-Direct Method}
\label{sec:framework}
\input{framework}

\section{Experiments}
\label{sec:results}
\input{results}

\section{Conclusion}
\label{sec:conclude}
\input{conclude}

%\vfill
%\pagebreak
%KK%
\vfill
\pagebreak
%KK%

\bibliographystyle{IEEEbib2}
\bibliography{ref2}

\appendix
\section{Proof of Proposition 1}
\label{sec:appendix-proof-prop1}
\input{appendix_proof_prop1.tex}

% \section{Proof of Equivalence of GDA algorithm}
% \label{sec:appendix}
% \input{appendix.tex}

\end{document}

%% file: defs.tex
\def\0{{\mathbf 0}}
\def\1{{\mathbf 1}}

\def\e{{\mathbf e}}

\def\u{{\mathbf u}}

\def\x{{\mathbf x}}
\def\y{{\mathbf y}}

\def\A{{\mathbf A}}
\def\B{{\mathbf B}}
\def\C{{\mathbf C}}
\def\D{{\mathbf D}}
\def\E{{\mathbf E}}

\def\H{{\mathbf H}}
\def\I{{\mathbf I}}

\def\L{{\mathbf L}}
\def\M{{\mathbf M}}

\def\P{{\mathbf P}}

\def\S{{\mathbf S}}

\def\U{{\mathbf U}}

\def\W{{\mathbf W}}

\def\ie{{\textit{i.e.}}}
\def\eg{{\textit{e.g.}}}

\def\cE{{\mathcal E}}

\def\cG{{\mathcal G}}

\def\cN{{\mathcal N}}
\def\cO{{\mathcal O}}

\def\cT{{\mathcal T}}

\def\cV{{\mathcal V}}

\def\bLambda{{\boldsymbol \Lambda}}

%% file: intro.tex
% (1) GSP on directed graph significance. (2) Research includes many aspects, sampling problem...

% (1) Existing sampling methods mainly focus on band-limited signal on undirected graph. (2) Though some works extend the sampling methods to a more general graph signal with smooth prior. Those methods can not be directedly applied to directed graph. Efficient (eigendecompose-free) methods band-limited-assumption-free sampling methods is required.

% Our contribution.

\textit{Graph signal processing} (GSP) extends traditional signal processing tools to analyze signals on irregular data kernels described by finite graphs \cite{ortega2018graph,cheung18}. 
Most existing GSP works consider undirected graph structures, where each edge connecting two nodes is bidirectional. 
However, \textit{directionality} plays an important role in many practical information dissemination scenarios \cite{marques2020signal}. 
For example, on Twitter, a celebrity often has a large following but personally follows very few users \cite{anger2011measuring}. 
Thus, in these scenarios it is critical to factor directionality into the network model, resulting in a directed graph.

\textit{Graph sampling} selects a node subset to collect samples, so that the target signal can be recovered in high fidelity \cite{tanaka2020sampling}.
Existing graph sampling works can be classified into two categories based on prior assumptions: 
i) an assumption on strict bandlimitedness of target signals with a cutoff frequency, and ii) a more general assumption assuming target signals are ``smooth" with respect to (w.r.t.) the underlying graph (\eg, more energy in low frequencies than high frequencies).
Bandlimited assumption in the first category means that a target signal lies strictly inside a linear subspace spanned by the first eigenvectors (Fourier modes) of a graph variation operator, such as the graph Laplacian matrix $\L$ or the adjacency matrix $\W$ \cite{chen2015discrete,tsitsvero2016signals,anis2014towards,anis2016efficient,wang2022mse}.
%Given chosen samples, the perfect recovery condition for graph sampling is derived in \cite{anis2014towards,chen2015discrete,tsitsvero2016signals}.
Assuming that the observed signal samples contain noise, \cite{chen2015discrete} proposed a greedy algorithm to select sample nodes under the E-optimality criterion \cite{pukelsheim2006optimal}, and \cite{tsitsvero2016signals} designed a greedy algorithm to minimize the reconstruction MSE.
To lower complexity, \cite{anis2016efficient} and \cite {sakiyama2019eigendecomposition} used graph spectral proxies and localization operators, respectively, to mitigate the computation burden of eigen-decomposition.
% Although these works can theoretically guarantee the performance of the reconstruction with the sampling nodes, they cannot be directly applied to directed graphs.
% The main reason is that the graph variation operator of directed graph can be asymmetric, their eigenspaces can have complex numbers and loses physical meanings.
% \red{Yuichi Tanaka has an eigendecomposition-free scheme. Cite that here.}

However, the strict bandlimited assumption of target signals is a strong one that many practical graph signals do not satisfy.
To relax this assumption, works in the second category assume that a target signal is generally smooth over a given graph \cite{chepuri2018graph,tanaka2020generalized}.
For example, \textit{graph Laplacian regularizer} (GLR) \cite{pang17}, \ie, $\x^\top \L \x$, is often used to quantify smoothness of signal $\x$ over a graph specified by Laplacian $\L$ \cite{chepuri2018graph,tanaka2020generalized,pang17,bai2020fast,dinesh2022point}.
GLR is often used to regularize under-determined signal reconstruction problems, such as denoising, dequantization, and interpolation  \cite{pang17,liu17,chen2021fast}.
Using GLR as signal prior, graph sampling based on \textit{Gershgorin disk alignment} (GDAS) was proposed to efficiently select sample nodes on undirected (signed) graphs under the E-optimality criterion \cite{bai2020fast,dinesh2022point}.
A key feature of GDAS is that it circumvents eigen-decomposition entirely and executes in linear time, and thus is scalable to large graphs. 

Although the above sampling algorithms are efficient and effective, they are all designed for undirected graphs, and cannot be easily applied to directed graphs.
One main challenge in directed graph sampling is the inherent difficulty in defining graph frequencies, due to the asymmetric nature of the directed graphs' variation operators, \eg, adjacency and Laplacian matrices, $\W$ and $\L$. 
Asymmetry means that the graph operator matrix may not be diagonalizable (and thus eigenvectors cannot be easily obtained), and even if it is, its eigenvalues can be complex, which are difficult to interpret (\eg, ordering of eigenvectors into frequencies from high to low is not obvious). 
% For directed graphs, their adjacency matrices and the corresponding Laplacian matrices are not symmetric.
% \red{we can say a lot more: frequencies on direct graphs are not well understood, partly due to the assymetric nature of graph variation operators like adjacency and Laplacian matrices that may not be diagonalizable.}
Though \cite{chen2015discrete,anis2016efficient} discussed in passing how their methods can be adapted to directed graphs, frequency and bandlimitedness notions are still not well understood on directed graphs.
% Therefore, the bandlimitedness and the smooth prior based on  are not suitable for directed graphs.

To circumvent the above challenge, in this work, we formulate a novel directed graph sampling problem using \textit{graph shift variation} (GSV) \cite{chen15} with a solution that completely avoids matrix asymmetry, and in so doing enable a variant of GDAS for fast sampling.
Specifically, we first define GSV $\|\L_{rw} \x\|^2_2$ as a smoothness prior for directed graph signal $\x \in \mathbb{R}^N$, where $\L_{rw}$ is a random-walk graph Laplacian for directed graph $\cG^d$. 
Using GSV as regularizer to reconstruct signal $\x$ from samples $\y = \H \x \in \mathbb{R}^K$, where $\H \in \{0,1\}^{K \times N}$ is a sampling matrix, the solution is $\x^* = \C^{-1} \H^\top \y$, with \textit{symmetric} coefficient matrix $\C = \H^\top \H + \mu \L_{rw}^\top \L_{rw}$. 
To minimize the worst-case reconstruction error (E-optimality), the sampling objective is to choose $\H$ to maximize the smallest eigenvalue $\lambda_{\min}(\C)$ of $\C$. 
To mitigate eigen-decomposition entirely, we devise a variant of previous GDAS to efficiently choose $\H$ to maximize a lower bound $\lambda^-_{\min}(\S\C\S^{-1})$---smallest Gershgorin disc left-end of a similarity transform of $\C$---based on \textit{Gershgorin circle theorem} (GCT) \cite{horn2012matrix}.
Experimental results show that our sampling method yields smaller signal reconstruction errors at a faster speed compared to competing schemes.
\textit{To the best of our knowledge, this is the first directed graph sampling algorithm in GSP free from explicit definitions of directed graph frequencies.}

%% file: prelim.tex
Consider a directed graph $\cG^d=(\cV, \cE, \W)$ with $N$ nodes $\cV$ and directed edges $\cE$. 
$\W$ is an adjacency matrix, where $W_{i,j} \in \mathbb{R}^+$ is the positive weight of directed edge $(i, j)$ if it exists in $\cE$.
% \red{$\langle,\rangle$ is typically used for inner product.}
We assume no self-loops, and thus $W_{i, i}= 0, \forall i$.
Denote by $\D$ the diagonal \textit{out-degree} matrix such that $D_{i,i} = \sum_{j}W_{i,j}$.
We assume that each node has strictly positive degree, \ie, $D_{i,i} > 0, \forall i$; this means that there are no sink nodes.
Graph Laplacian matrix of the directed graph is defined as $\L \triangleq \D - \W$.
The normalized adjacency matrix is $\bar{\W} = \D^{-1}\W$, and the random-walk graph Laplacian is $\L_{rw} \triangleq \D^{-1} \L = \I - \bar{\W}$.
Finally, we assume that there exists at least one node $v$ such that there are directed paths from all other nodes $v^\prime \in \cV$ to node $v$.
This assumption ensures that the rank of the random-walk Laplacian matrix $\L_{rw}$ is $N-1$ \cite{veerman2019diffusion}.

%% file: problem.tex
We first review a previously proposed smoothness prior----\textit{graph shift variation} (GSV) \cite{chen15}----and use it to reconstruct a directed graph signal given limited samples. 
We then formulate a directed graph sampling problem given a defined signal reconstruction scheme.

\subsection{Signal Reconstruction on a Directed Graph}
\label{sec:sig-recon}

\textbf{Graph Shift Variation Prior}. 
Denote by $\x\in \mathbb{R}^{N}$ a signal on a directed graph $\cG^d$. 
An important assumption in GSP is that the signal is smooth w.r.t. an underlying graph. 
For undirected graphs, there exist different smoothness measures of a graph signal, such as GLR \cite{pang17} and \textit{graph total variation} (GTV) \cite{bai19}.
%On undirected graphs, the quadratic form $\x^\top \L \x$ is often adopted as the measure of smoothness of signal $\x$.
However, for directed graphs, because Laplacian $\L$ is asymmetric, smoothness priors like GLR cannot be used directly.
In this paper, we adopt GSV in \cite{chen15} as the smoothness measure of a signal $\x$ on directed graph $\cG^d$, \ie, 
\begin{align}
S(\x) = \left\Vert \x - \frac{1}{\vert \lambda^a_{\max}(\W_s) \vert}\W_s \x \right\Vert_2^2 .
\end{align} 
Here, $\W_s$ is a \textit{graph shifting operator} \cite{sandryhaila2013discrete,sandryhaila2014discrete}, and it has the same support as adjacency matrix $\W$. $\lambda^a_{\max}(\W_s)$ is the largest magnitude eigenvalue of $\W_s$, and $|\lambda^a_{\max}(\W_s)|$ is the \textit{spectral radius} of matrix $\W_s$. $1/|\lambda^a_{\max}(\W_s)|$ is used for normalization.
$\frac{1}{\vert \lambda^a_{\max}(\W_s) \vert}\W_s \x$ shifts each node's sample to its one-hop neighbors, and $S(\x)$ measures the difference between signal $\x$ and its shifted version.
% \red{your $\W$ is already row-stochastic, right? So $\L_{rw} = \I - \W$ is a similar transform of normalized Laplacian $\cL = \D^{1/2} \L_{rw} \D^{-1/2}$, which has eigenvalues $[0,2]$. So $\L_{rw}$ also has eigenvalues $[0,2]$, $\W$ has eigenvalues $[-1,1]$, and the spectral radius of $\W$ is $1$. Why do we still need $\lambda_{\max}(\W)$?}
% where $\lambda_{\max}(\W)$ is the largest magnitude eigenvalue (spectral radius) of the adjacency matrix $\W$. 
% Note that $\W$ is also called the graph shifting operator \cite{sandryhaila2013discrete,sandryhaila2014discrete}, and $\W\x$ shifts each node's original signal to its one-hop neighbors. 
% Thus, the total variation in (\ref{eqn:ori-tv-digraph}) measures the difference between the original graph signal and it shifted version.

In this work, we use the normalized adjacency matrix $\bar{\W} = \D^{-1}\W$ as the graph shift operator, since its largest eigenvalue is $\lambda^a_{\max}(\bar{\W}) = 1$. Consequently, the GSV prior is defined as
\begin{equation}
\label{eqn:gsv-digraph}
S(\x) = \left\Vert \x - \bar{\W} \x \right\Vert_2^2 = \left\Vert \L_{rw}\x \right\Vert_2^2 = \x^\top \L_{rw}^\top \L_{rw} \x.
\end{equation}
This GSV prior in \eqref{eqn:gsv-digraph} is similar to the \textit{left eigenvector random walk graph Laplacian} (LeRAG) regularizer in \cite{liu17}. 
It is shown that the smooth prior in (\ref{eqn:gsv-digraph}) is insensitive to vertex degrees.
Further, it is shown \cite{liu17} that GSV of a constant signal $\x = c \1$ evaluates to $S(\x) = 0$, which is intuitive and important for imaging applications.

% As we assume that the adjacency matrix $\W$ is row-stochastic, its largest eigenvalue is $\lambda_{\max}(\W) = 1$ \red{right}. 
% Furthermore, as the graph Laplacian matrix is $\L = \I - \W$, we have
% \begin{equation}
% \label{eqn:tv-digraph}
% S(\x) = \left\Vert (\I - \W)\x \right\Vert_2^2 = \left\Vert \L\x \right\Vert_2^2 = \x^\top \L^\top \L \x.
% \end{equation}
% \red{should call it the random-walk graph Laplacian.}
% \red{this looks similar to LeRAG in Xianming's TIP'17.}
% \blue{(Yes. I read Xianming's TIP'17 paper, and add the following short discussion.)}

\vspace{0.05in}
\noindent
\textbf{Signal Reconstruction using GSV Prior}. 
Suppose that we obtain $K$ samples, $\y \in \mathbb{R}^{K}$, of graph signal $\x \in \mathbb{R}^{N}$, where $K < N$. 
We aim to reconstruct signal $\x^*$ given observation $\y$.
% \red{I am not fond of this notation, because implicitly, $\cT$ is related to $\H$, while in equation (3) they look unrelated.}
To regularize this under-determined problem, we employ GSV \eqref{eqn:gsv-digraph} as prior and solve the following regularized optimization problem \cite{chen15,sandryhaila2013classification}:
\begin{equation}
\label{eqn:recon-opt}
\x^* = \arg\min_{\x} \Vert \H\x - \y \Vert_2^2 + \mu \, \x^\top \L_{rw}^\top \L_{rw} \x,
\end{equation}
where $\H \in \{0, 1\}^{K \times N}$ is the sampling matrix, and $\mu > 0$ is a weight parameter that trades off the fidelity term with the GSV prior. 
%The first term in (\ref{eqn:recon-opt}) is the fidelity term, which is the reconstruction error on the observed nodes, while the second term in (\ref{eqn:recon-opt}) is the smooth prior regularizer weighted by a constant $\mu > 0$. 
The optimal solution $\x^*$ to \eqref{eqn:recon-opt}, which is quadratic and convex, can be obtained by solving the following linear system
\begin{equation}
    \label{eqn:recon-lin-sys}
    \left( \H^\top \H + \mu \L_{rw}^\top \L_{rw}  \right) \x^* = \H^\top \y.
\end{equation}
% \red{Yuanchao's TSP'20 has a similar proof, I think. A simple proof can be that given both matrices are PSD, a vector $\x$ cannot be in the null space of $\H^\top \H$ and $\L^\top \L$ at the same time. It just a couple of lines.}
Note that both $\H^\top \H$ and $\L_{rw}^\top \L_{rw}$ are positive semi-definite matrix.
For matrix $\L_{rw}$, we have $\L_{rw}\1 = (\I - \D^{-1}\W) \1 = \1 - \1 = \0$. Thus, $Span\{\1\} \subseteq Null(\L_{rw})$.
With the assumption in Sec.\;\ref{sec:prelim} that there exists at least one node that can reach any other nodes through directed paths, from Theorem 4.5 in \cite{veerman2019diffusion}, the dimension of $Null(\L_{rw})$ is 1.
Thus, we have $Null(\L_{rw}) = Span\{ \1\}$.
Note that for $\x \notin Span\{ \1\}$, we have $\x^\top \H^\top \H \x \ge 0$ and $\x^\top \L_{rw}^\top \L_{rw} \x > 0$.
% Thus, we have $\z^\top (\H^\top \H + \L_{rw}^\top \L_{rw}) \z > 0$.
For $\x \in Span\{ \1 \}$, say $\x = c\1$ where $c$ is a non-zero real scalar, we have $\x^\top \H^\top \H \x > 0$ and $\x^\top \L_{rw}^\top \L_{rw} \x \ge 0$.
% Therefore, we also have $\x^\top (\H^\top \H + \L_{rw}^\top \L_{rw}) \x > 0$.
In summary, for any $\x \in \mathbb{R}^{N}$, we have $\x^\top (\H^\top \H + \L_{rw}^\top \L_{rw}) \x > 0$, and thus, $(\H^\top \H + \L_{rw}^\top \L_{rw})$ is a positive definite matrix and invertible.
Consequently, the unique optimal solution to (\ref{eqn:recon-opt}) as well as (\ref{eqn:recon-lin-sys}) is
\begin{equation}
    \label{eqn:recon}
    \x^* = \left( \H^\top \H + \mu \L_{rw}^\top \L_{rw}  \right)^{-1} \H^\top \y.
\end{equation}
Note that given the coefficient matrix in (\ref{eqn:recon-lin-sys}) is symmetric, sparse, and positive definite, $\x^*$ can be solved using \textit{conjugate gradient} (CG) \cite{hestenes1952methods} without performing any matrix inverse.

\subsection{Directed Graph Sampling Problem}
\label{sec:sample}

Observation $\y$ may contain noise. 
Given a sampling budget $K$, to minimize worst-case reconstruction error using reconstruction \eqref{eqn:recon}, we adopt the E-optimality criterion \cite{chen2015discrete,bai2020fast} to maximize the smallest eigenvalue of coefficient matrix $\H^\top \H + \mu \L_{rw}^\top \L_{rw}$.
For notation simplicity, we define diagonal matrix $\A \triangleq \H^\top \H$, whose diagonal entry $A_{i, i} = 1$ if node $i$ is sampled and $A_{i, i} = 0$ otherwise.
The sampling problem is thus formulated as
\begin{align}
\max_{\A} & ~~ \lambda_{\min} (\A + \mu \L_{rw}^\top \L_{rw}) 
\nonumber \\
\mbox{s.t.} & ~~ A_{i,i} \in \{1, 0\},  \forall i, ~~~~~~~ \text{tr}(\A) = K.
\label{eqn:ori-prob}
\end{align}
The second constraint in (\ref{eqn:ori-prob}) indicates that we can sample $K$ nodes. 
Note that optimization \eqref{eqn:ori-prob} is combinatorial in nature and NP-hard in general. 
Next, we develop an efficient algorithm for \eqref{eqn:ori-prob}.

%% file: framework.tex
We first review the Gershgorin disc alignment sampling (GDAS) algorithm in \cite{bai2020fast}. 
We then derive a lower bound of the objective in \eqref{eqn:ori-prob} and design an efficient algorithm to solve \eqref{eqn:ori-prob} based on GDAS. 

\subsection{Gershgorin Disc Alignment Algorithm}
\label{sec:gda-alg}

The foundation of GDAS is Gershgorin Circle Theorem (GCT) \cite{horn2012matrix}. 
\textit{Gershgorin disc} $\Psi_i$ of the $i$-th row of a real matrix $\M$ is a circle on the complex plane, with center $(M_{i,i}, 0)$ and radius $r_i = \sum_{j\neq i} \vert M_{i,j} \vert$.
GCT states that all eigenvalues of $\M$ reside inside the union of Gershgorin discs of $\M$. 
For a real symmetric matrix $\M$ whose eigenvalues are also real, we can define the lower bound of $\lambda_{\min}(\M)$ as the smallest left-end of all discs:
\begin{equation}    
\lambda^-_{\min}(\M) \triangleq \min_i \{M_{i,i} - r_i\} \le \lambda_{\min}(\M).
\label{eqn:lbd-eig-min}
\end{equation}
When applying a similarity transform on $\M$ with an invertible diagonal matrix $\S = \text{diag}(s_1, \ldots, s_N), s_i \neq 0, \forall i$, we obtain a new matrix $\S \M \S^{-1}$ with the same eigenvalues as original $\M$. 
Thus, 
\begin{equation}
\label{eqn:gda-lb}
\lambda^-_{\min}(\S\M\S^{-1}) \le \lambda_{\min}(\S\M\S^{-1}) = \lambda_{\min}(\M).
\end{equation}

To apply GDAS \cite{bai2020fast} to approximately solve the sampling problem \eqref{eqn:ori-prob}, we employ the lower bound \eqref{eqn:gda-lb} and relax the objective in \eqref{eqn:ori-prob} to $\lambda^-_{\min}(\S (\A + \mu \L_{rw}^\top \L_{rw}) \S^{-1}) = \lambda^-_{\min}(\A + \mu \S \L_{rw}^\top \L_{rw} \S^{-1})$. 
Thus, the optimization problem is reformulated as
\begin{align}
\max_{\A,\S} &~~ \lambda^-_{\min}(\A + \mu \S \L_{rw}^\top \L_{rw} \S^{-1})
\nonumber \\
\mbox{s.t.} &~~ A_{i,i} \in \{1, 0\}, \forall i, ~~~~~~ \text{tr}(\A) = K.
\label{eqn:gda-prob}
\end{align}
GDAS efficiently solves the relaxed problem \eqref{eqn:gda-prob} through \textit{disc shifting} for optimal $\A^*$ and \textit{disc scaling} for optimal $\S^*$. We refer readers to \cite{bai2020fast} for details of the algorithm.

\vspace{0.05in}
\noindent
\textbf{Graph Balancing Algorithm for \eqref{eqn:gda-prob}}. 
GDAS requires that all $\L_{rw}^\top \L_{rw}$'s disc left-ends to be initially aligned at the same value before sampling.
However, this is not always satisfied for random-walk Laplacian $\L_{rw}$ of a directed graph.
Define $\tilde{\L} \triangleq \L_{rw}^\top \L_{rw}$.
We see that $\tilde{\L} \1 = \L_{rw}^\top \L_{rw}\1 = \0$, \ie, for any row $i$ of $\tilde{\L}$, $\tilde{L}_{i, i} = -\sum_{j\neq i} \tilde{L}_{i, j}$.
This means that $\tilde{\L}$ is a Laplacian matrix for an undirected graph without self-loops, and its corresponding adjacency matrix is $\tilde{\W} = \text{Diag}(\tilde{\L}) - \tilde{\L}$.
Hence, if there exist negative weights in $\tilde{\W}$, then  $\tilde{\L}$'s disc left-ends are not all aligned at $0$, since $c_i - r_i = \tilde{L}_{i,i} - \sum_{j\neq i} |\tilde{L}_{i,j}| = \sum_{j\neq i} (-\tilde{L}_{i,j} - |\tilde{L}_{i,j}|) \neq 0$.  

To align disc left-ends, one method is to first \textit{balance}\footnote{A graph is balanced if there are no cycles of odd number of negative edges \cite{dinesh2022point}.} the graph using an algorithm in \cite{dinesh2022point}, then align the discs' left-ends of the Laplacian $\tilde{\L}^b$ of the balanced graph $\cG^b$ via a similarity transform $\tilde{\L}^p = \tilde{\S} \tilde{\L}^b \tilde{\S}^{-1}$, where $\tilde{\S} = \text{diag}(\tilde{s}_1, \ldots, \tilde{s}_n)$ and $\{\tilde{s}_i\}$ are computed from the fist eigenvector of $\tilde{\L}^b$  \cite{yang22}.
Subsequently, GDAS \cite{bai2020fast} can be employed for sampling by maximizing $\lambda^-_{\min}(\A + \S \tilde{\L}^p \S^{-1})$.
We refer to this method as \texttt{GDA-Balance}.

\subsection{Analysis of the Lower Bound of the Objective Function}

% \purple{With the lower bound in (\ref{eqn:gda-lb}), we relax the original sampling problem in  (\ref{eqn:ori-prob}) by replacing $\M$ in (\ref{eqn:gda-prob}) with $\mu\L^\top\L$.
% In this case, we aim to maximize the lower bound of the objective in (\ref{eqn:ori-prob}), \ie, $\lambda^-_{\min}(\A + \mu\S \L^\top\L \S^{-1})$.
% However, we cannot directly apply the GDA algorithm to solve the relaxed problem, because it requires the left ends of $\M$'s discs to be aligned.
% In general, $\mu \L^\top \L$ does not satisfy this requirement.
% One possible solution is to align the left ends of the discs of $\mu\L^\top \L$ with the algorithm in \cite{dinesh2022point}, and then apply the GDA algorithm.
% The main challenge of using the algorithm in \cite{dinesh2022point} for preprocessing is that this algorithm needs to go through the nodes in the graph serially, and it can be time-consuming on large-size graphs.}
% \red{does $\L^\top \L$ correspond to a generalized graph Laplacian matrix of a balanced signed graph? I suspect that all the off-diagonal terms of $\L^\top \L$ remains non-positive, which means $\L^\top \L$ is a generalized Laplacian for a positive graph and thus requires no graph balancing. can you check? if not, u need to balance the graph first, then compute the 1st eigenvector to align its left-ends. graph balancing takes an extra step, but it's still linear time. we should have this at least as a comparison method.}

Although \texttt{GDA-Balance} can solve the relaxed sampling problem \eqref{eqn:gda-prob}, it has drawbacks.
%it does not exploit the unique properties of the random-walk graph Laplacian $\L_{rw}$.
First, the graph balancing procedure in \cite{dinesh2022point} can be computation-expensive.
Second, the resulting balanced graph Laplacian $\tilde{\L}^b$ may not be PSD, while $\L_{rw}^\top \L_{rw}$ is PSD with $\lambda_{\min}(\L_{rw}^\top \L_{rw}) = 0$, as discussed in Sec.\;\ref{sec:sig-recon}.
This means that computed lower bound using \texttt{GDA-Balance} for objective $\lambda_{\min}(\A + \mu \L_{rw}^\top \L_{rw})$ may be loose in practice. 
%\blue{Specifically, \texttt{GDA-Balance} is designed for signed graph. Therefore, after the balancing process, the smallest eigenvalue of $\tilde{\L}^b$ as well as $\tilde{\L}^p$ and the relaxed lower bound $\lambda^-_{\min}(\A + \S \tilde{\L}^p \S^{-1})$ can be negative.However, in Sec.\;\ref{sec:sig-recon}, we have shown that $\tilde{\L} = \L_{rw}^\top \L_{rw}$ is a positive semi-definite matrix, and the objective $\lambda_{\min}(\A + \mu \L_{rw}^\top \L_{rw})$ is always positive, even if $\tilde{\L}$ is the generalized Laplacian matrix of a signed graph. This shows that the balancing process of \texttt{GDA-Balance} may result in a loose lower bound of the objective $\lambda_{\min}(\A + \mu \L_{rw}^\top \L_{rw})$.}

From the definition of $\L_{rw}$ of the directed graph in Sec.\;\ref{sec:prelim}, we see that the left-ends of $\L_{rw}$'s discs are aligned at $(0, 0)$ on the complex plane, though it is not symmetric.
This observation motivates us to relax the objective $\lambda_{\min}(\A + \mu \L_{rw}^\top \L_{rw})$ in (\ref{eqn:ori-prob}) with a lower bound in the form of $\lambda^-_{\min}(\S(\delta \A + \rho \L_{rw})\S^{-1})$, where $\delta$ and $\rho$ are positive constants.
Doing so means we can directly apply GDAS on the relaxed lower bound due to the alignment of the left ends of $\L_{rw}$.
In the following, we first derive this lower bound, and then efficiently calculate the parameters $\delta$ and $\rho$.
% This indicates that we can directly apply the GDA algorithm to the objective $\lambda^-_{\min}(\S(\delta \A + \rho \L_{rw})\S^{-1})$, where $\delta$ and $\rho$ are positive constant.
% Thus, in the following, we derive a lower bound of $\lambda_{\min}(\A + \mu \L_{rw}^\top \L_{rw})$ with $\lambda^-_{\min}(\S(\delta \A + \rho \L_{rw})\S^{-1})$, and efficiently calculate the two parameters $\delta$ and $\rho$.

% To apply the GDA algorithm, we first find a lower bound related to the objective in (\ref{eqn:ori-prob}).
\begin{proposition}
\label{prop:main-lb}
Given a normalized adjacency matrix $\bar{\W}$, its random-walk graph Laplacian $\L_{rw}$, sampling budget $K$, and a non-negative hyper-parameter $\mu > 0$, for any positive scalar $0 < \epsilon < 1$, we define two sets of parameters $\delta$ and $\rho$ depending on $\mu$ as follows:
\begin{align}
    \bullet\quad &\delta = \sqrt{1 - \epsilon \mu},~ \rho = \frac{-\delta c + \sqrt{\delta^2 c^2 + 4\mu}}{2}, \text{ where,}\label{eqn:param-delta-rho}\\
    &c = \frac{3 + \max_{i} \sum_{j=1}^{K-1} \bar{W}_{[j], i}}{
    \min\limits_{\substack{\A: \text{tr}(\A) = K \\ A_{i,i} = \{0, 1\}}} \lambda_{\min}(\L_{rw}^\top \L_{rw} + \epsilon \cdot \A)} \text{,  for  } \mu \le 1;\text{  and}\label{eqn:param-c}\\
    \bullet\quad &\delta = \frac{-\rho c + \sqrt{\rho^2 c^2 + 4}}{2},~ \rho = \sqrt{\mu - \epsilon}, \text{ where,}\label{eqn:param-delta-rho-mu-ge1}\\
    &c = \frac{3 + \max_{i} \sum_{j=1}^{K-1} \bar{W}_{[j], i}}{\min\limits_{\substack{\A: \text{tr}(\A) = K \\ A_{i,i} = \{0, 1\}}} \lambda_{\min}(\A + \epsilon\cdot\L_{rw}^\top \L_{rw})} \text{,  for  } \mu > 1.\label{eqn:param-c-mu-ge1}
\end{align}
$\bar{W}_{[j], i}$ is the $j$-th largest element in the $i$-th column of $\bar{\W}$.
With the defined parameters $\delta$ and $\rho$ above, if a non-negative invertible $\S$ satisfies $\lambda^-_{\min}(\S(\delta\A + \rho \L_{rw})\S^{-1}) \ge 0$, then the following inequality holds
\begin{equation}
    \frac{(\lambda^-_{\min}(\S(\delta\A + \rho \L_{rw})\S^{-1}))^2}{\gamma^2_{\max}} \le \lambda_{\min} (\A + \mu \L_{rw}^\top \L_{rw}),\label{eqn:1}
\end{equation}
where the positive constant $\gamma_{\max}$ is defined as
$$
\begin{gathered}
    \gamma_{\max} \triangleq \max_{\A} \gamma(\P)\\
     s.t.~~~ \delta\A + \rho\L_{rw} = \P \bLambda \P^{-1},~~~ A_{i,i}=\{0,1\},~~~ tr(\A) = K.
\end{gathered}
$$
In the above, $\gamma(\P)$ is the condition number of matrix $\P$.

\end{proposition}

The proof of Proposition \ref{prop:main-lb} is in Appendix \ref{sec:appendix-proof-prop1}.
% From Proposition \ref{prop:main-lb}, we derive a lower bound of the objective $\lambda_{\min} (\A + \mu \L_{rw}^\top \L_{rw})$ in (\ref{eqn:ori-prob}).
With Proposition \ref{prop:main-lb}, we can relax the objective $\A + \mu \L_{rw}^\top \L_{rw}$ with its lower bound in (\ref{eqn:1}).
Compared to \texttt{GDA-Balance}, we see that the derived lower bound in \eqref{eqn:1} is always positive, and thus is tighter than the lower bound given by \texttt{GDA-Balance}.

Note that $\gamma_{\max}$ is a constant that does not depend on $\A$ or $\S$.
%Thus, we can maximize $(\lambda^-_{\min}(\S(\delta\A + \rho \L_{rw})\S^{-1}))^2$ instead.
Further, since we assumed $\lambda^-_{\min}(\S(\delta\A + \rho \L_{rw})\S^{-1}) \ge 0$ in Proposition \ref{prop:main-lb}, we can maximize $\lambda^-_{\min}(\S(\delta\A + \rho \L_{rw})\S^{-1})$, where we can directly apply GDAS.
Note that during the optimization process in GDAS, it naturally ensures that $\lambda^-_{\min}(\S(\delta\A + \rho \L_{rw})\S^{-1}) \ge 0$ \cite{bai2020fast}.
Thus, this assumption in Proposition \ref{prop:main-lb} always holds when we use GDAS on $\lambda^-_{\min}(\S(\delta\A + \rho \L_{rw})\S^{-1})$.

\begin{algorithm}[t]
    \SetAlgoLined
    \KwIn{Normalized adjacency matrix $\bar{\W}$, Laplacian matrix $\L_{rw}$, priors weight $\mu$, sampling budget $K$, and total number of nodes $N$, parameter $\epsilon$.}
    \KwOut{Optimal $\A^*$}
        Calculate numerator of $c$, that is, $3 + \max_{i} \sum_{j=1}^{k-1} W_{[j], i}$\;

        \eIf{$\mu \le 1$}{
            % Calculate second smallest eigenvalue $\lambda_2$ of $\L_{rw}^\top \L_{rw}$ with LOBPCG\;

            % Decide parameter $\epsilon$ such that $\epsilon \le \frac{\lambda_2 N}{k}$\;
    
            Estimate the denominator $c$ with $\frac{k\epsilon}{N}$\;
    
            Calculate $c$, $\delta$ and $\rho$ with (\ref{eqn:param-c}) and (\ref{eqn:param-delta-rho})\;
        }{
            Estimate the denominator $c$ with (\ref{eqn:approx-denom-c-mu-ge1})\;

            Calculate $c$, $\delta$ and $\rho$ with (\ref{eqn:param-c-mu-ge1}) and (\ref{eqn:param-delta-rho-mu-ge1})\;
        }

        Use GDA algorithm \cite{bai2020fast} on $\lambda^-_{\min}(\S(\delta\A + \rho \L_{rw})\S^{-1})$ for optimal $\A^*$\;
        \caption{The \texttt{GDA-Direct} Sampling Algorithm on Directed Graph}
    \label{alg:di-gda}
  \end{algorithm}

\subsection{GDA-Direct Algorithm for Sampling on Directed Graphs}

To maximize $\lambda^-_{\min}(\S(\delta\A + \rho \L_{rw})\S^{-1})$ using GDAS, we first need to calculate the parameters $c$ in (\ref{eqn:param-c}) and (\ref{eqn:param-c-mu-ge1}), and $\delta$ as well as $\rho$ in (\ref{eqn:param-delta-rho}) and (\ref{eqn:param-delta-rho-mu-ge1}).
We discuss the efficient calculation / approximation of these parameters and elaborate the complete \texttt{GDA-Direct} sampling algorithm.
The algorithm is summarized in Algorithm\;\ref{alg:di-gda}.

To apply GDAS on $\lambda^-_{\min}(\S(\delta\A + \rho \L_{rw})\S^{-1})$, we need to determine the parameters $\delta$ and $\rho$, which depends on the parameter $c$ in (\ref{eqn:param-c}) and (\ref{eqn:param-c-mu-ge1}).
To calculate $c$, we need to calculate its numerator $3 + \max_{i} \sum_{j=1}^{K-1} W_{[j], i}$ and denominator $\min_{\A} \lambda_{\min}(\L_{rw}^\top \L_{rw} + \epsilon \cdot \A)$ or $\min_{\A} \lambda_{\min}(\A + \epsilon\L_{rw}^\top \L_{rw})$, respectively.
To calculate the numerator, for the $i$-th column of matrix $\W$, we find the top $K-1$ largest entries, \ie, $W_{[j], i}$ for $j=1,\ldots, K-1$.
Then, with the computed $\sum_{j=1}^{K-1} W_{[j], i}$ for each column, we calculate the numerator (line 1 in Algorithm \ref{alg:di-gda}).
The time complexity for this step is $\cO(N d^{in}_M \log(K-1))$, where $d^{in}_{M}$ is the largest in-degree of the directed graph.
Here, we can see that the time complexity is proportional to the graph size $N$.
% Although the time complexity is proportional to $N^2$, this step is still efficient, because it only includes sorting operation instead of intensive matrix operation.
% \red{not sure this is a good reason,}

Next, we calculate the denominator $\min_{\A} \lambda_{\min}(\L_{rw}^\top \L_{rw} + \epsilon \cdot \A)$, when $\mu \le 1$.
Here, optimization variable $\A$ is diagonal and satisfies $A_{i, i} = \{ 1, 0 \}$ and $\text{tr}\{ \A \} = K$. 
%This denominator is similar to the objective in (\ref{eqn:ori-prob}), and thus, hard to  calculate directly.
To simplify calculation, we assume small positive $\epsilon$ near $0$, and $\L_{rw}^\top \L_{rw} + \epsilon \cdot \A$ is a perturbation of $\L_{rw}^\top \L_{rw}$.
Denote by $0 = \lambda_1 < \lambda_2 \le \cdots \le \lambda_N $ the eigenvalues of $\L_{rw}^\top \L_{rw}$, and $\u_i$ the corresponding eigenvector of $\lambda_i$.
According to the analysis of the perturbed matrix's eigenvalues \cite{ceci2018small,murthy1988approximations,wilkinson1988algebraic}, we can approximate the eigenvalue of $\L_{rw}^\top \L_{rw} + \epsilon \cdot \A$ as
\begin{equation}
\tilde{\lambda}_i \approx \lambda_i + \epsilon \u_i^\top \A \u_i.
\end{equation}
Note that the smallest eigenvalue of $\L_{rw}^\top\L_{rw}$ and its corresponding eigenvector are $\lambda_1 = 0$ and $\u_1 = \frac{\1}{\sqrt{N}}$, respectively. 
Thus, approximated eigenvalue $\tilde{\lambda}_1$ is
\begin{equation}
    \tilde{\lambda}_1 \approx 0 + \epsilon \times \frac{1}{N} \1^\top \A \1 = \frac{K\epsilon}{N}.
\end{equation}

For other approximated eigenvalues, we have $\tilde{\lambda}_i \approx \lambda_i + \epsilon \u_i^\top \A \u_i \ge \lambda_i$ for $i > 1$.
To ensure that $\tilde{\lambda}_1$ is the smallest approximated eigenvalue, we only need $\tilde{\lambda}_1 \le \lambda_2$, that is,
\begin{equation}
    \label{eqn:epsilon}
    \epsilon \le \frac{\lambda_2 N}{K}.
\end{equation}
Note that \eqref{eqn:epsilon} also gives us a guide on how to choose $\epsilon$ when $\mu \le 1$.
For the second smallest eigenvalue of $\L_{rw}^\top\L_{rw}$, \ie, $\lambda_2$, we can use LOBPCG \cite{knyazev2001toward} to efficiently calculate it.
When $\epsilon$ satisfies the condition in (\ref{eqn:epsilon}), we can approximate the denominator of $c$ as $\lambda_{\min}(\L_{rw}^\top \L_{rw} + \epsilon \cdot \A) \approx \tilde{\lambda}_1 = \frac{K\epsilon}{N}$, which is independent of variable $\A$.
Further, we can calculate $\delta$ and $\rho$ with the approximated $c$ and (\ref{eqn:param-delta-rho}) (line 3 \& 4 in Algorithm \ref{alg:di-gda}).

Similarly, when $\mu > 1$, we can also approximate the denominator of $c$ in \eqref{eqn:param-c-mu-ge1}, \ie, $\min_{\A} \lambda_{\min}(\A + \epsilon\cdot\L_{rw}^\top \L_{rw})$.
Since matrix $\A$ has $N-K$ zero eigenvalues whose corresponding eigenvectors are canonical basis vectors $\e_i$, the smallest eigenvalue of $\A + \epsilon\cdot\L_{rw}^\top \L_{rw}$ can be approximated as
\begin{equation}
    \tilde{\lambda}_1 \approx 0 + \epsilon \e_i^\top \L_{rw}^\top \L_{rw} \e_i.
\end{equation}
Note that since we assumed there are no sink nodes in the graph, we have $(L_{rw})_{i, i} = 1$ for all $i$, and $\e_i^\top \L_{rw}^\top \L_{rw} \e_i \ge 1$. Therefore, when $\mu > 1$, the denominator of $c$ can be approximated as
\begin{equation}
    \min_{\A} \lambda_{\min}(\A + \epsilon\cdot\L_{rw}^\top \L_{rw})\approx \min \tilde{\lambda}_1 = \epsilon \cdot \min_{i}\e_i^\top \L_{rw}^\top \L_{rw} \e_i.
    \label{eqn:approx-denom-c-mu-ge1}
\end{equation}
The time complexity of calculating $c$, $\delta$, and $\rho$ are $\cO(1)$ for $\mu \le 1$ and $\cO(N)$ for $\mu > 1$, respectively.
After determining $\delta$ and $\rho$, we can employ GDAS to maximize $\lambda^-_{\min}(\S(\delta\A + \rho \L_{rw})\S^{-1})$ to obtain the optimal solution $\A^*$ (line 9 in Algorithm \ref{alg:di-gda}).
The time complexity of this step is $\cO(N)$ \cite{bai2020fast}.
In summary, the overall time complexity is $\cO(N\cdot (1 + d^{in}_{M} \log(K-1)))$.
When the graph is sparse, and the sampling budget $K$ is small, the overall time complexity is roughly $\cO(N)$.

%% file: results.tex
\begin{figure*}[!t]
\centering
\subfigure{\includegraphics[width=0.31\textwidth]{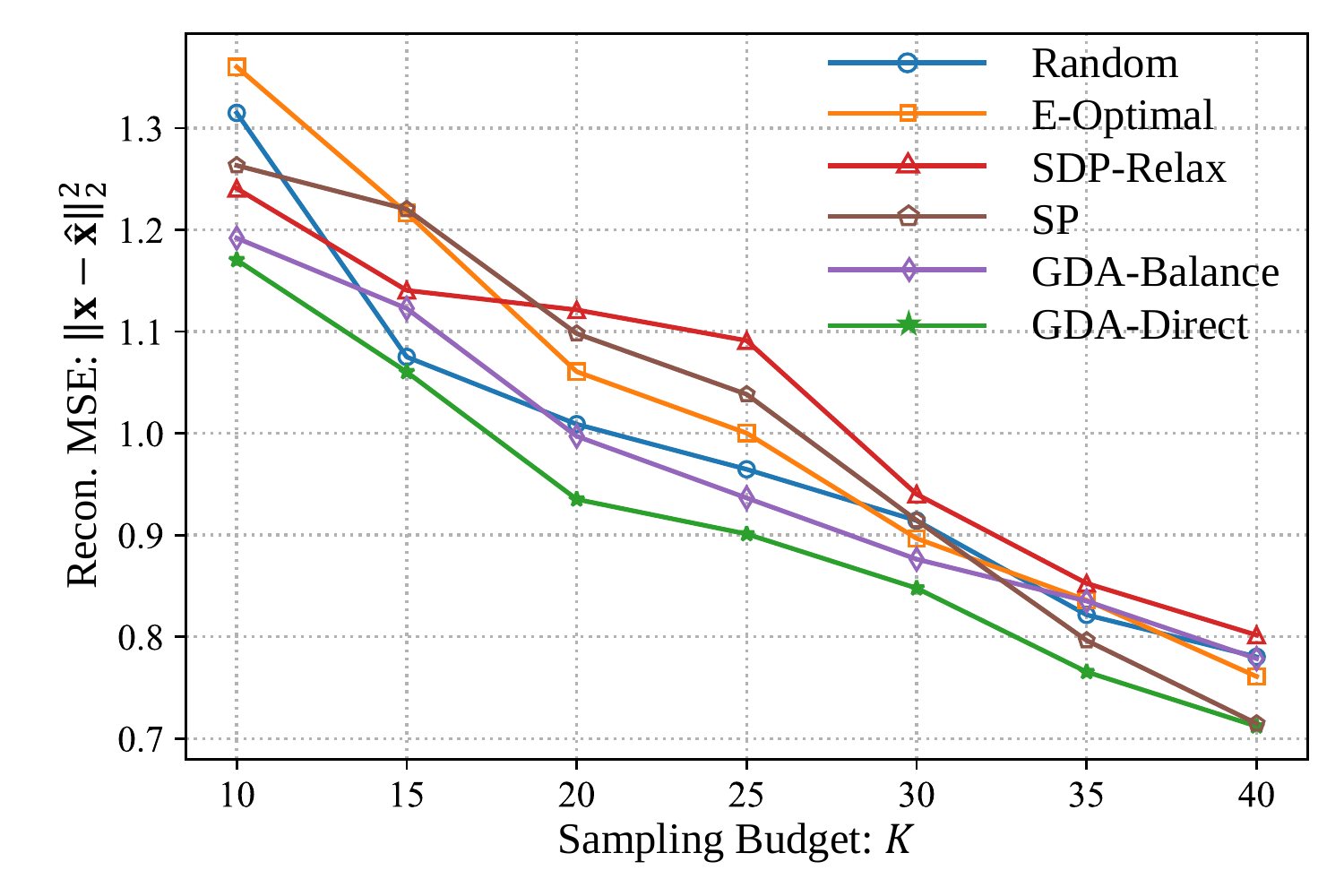}}
\subfigure{\includegraphics[width=0.31\textwidth]{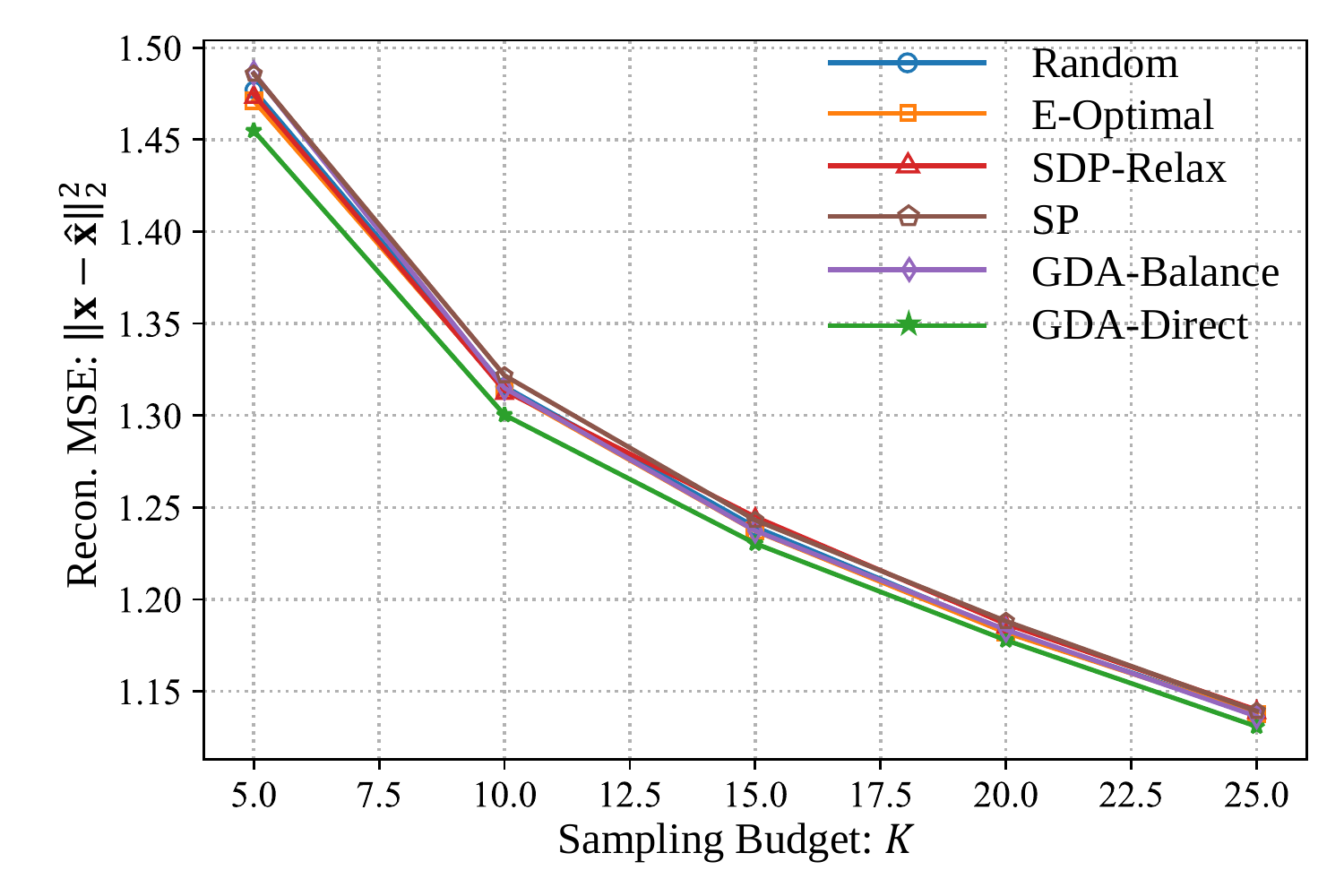}}
\subfigure{\includegraphics[width=0.31\textwidth]{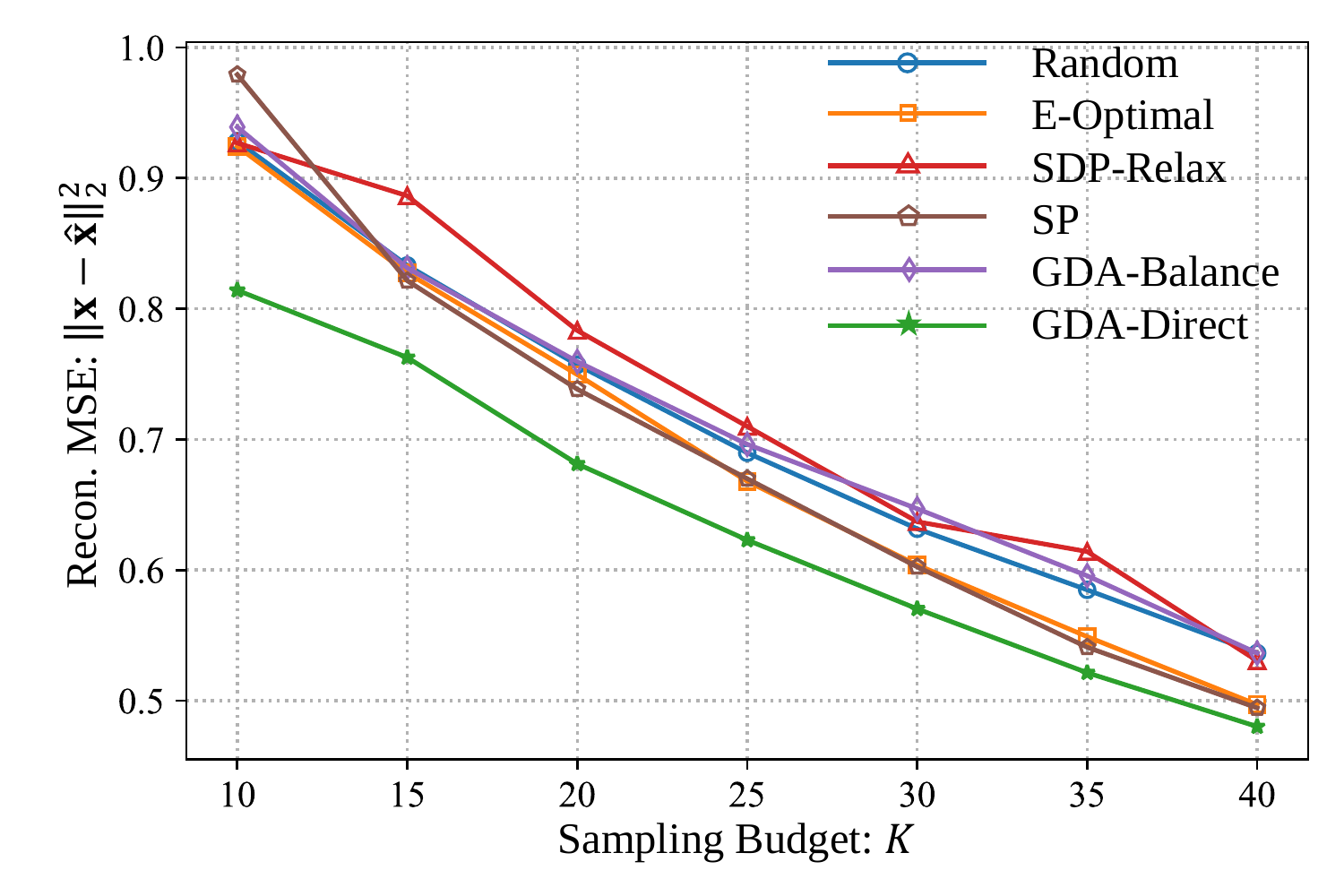}}
\caption{Reconstruction MSE of (left) GS1, (mid) GS2, and (right) GS3 on random graph ($N = 200$ and $p = 0.1$).}
\label{fig:mse-results-N200}
\end{figure*}

% ----- to be removed -----
\begin{figure*}[!t]
\centering
\subfigure{\includegraphics[width=0.31\textwidth]{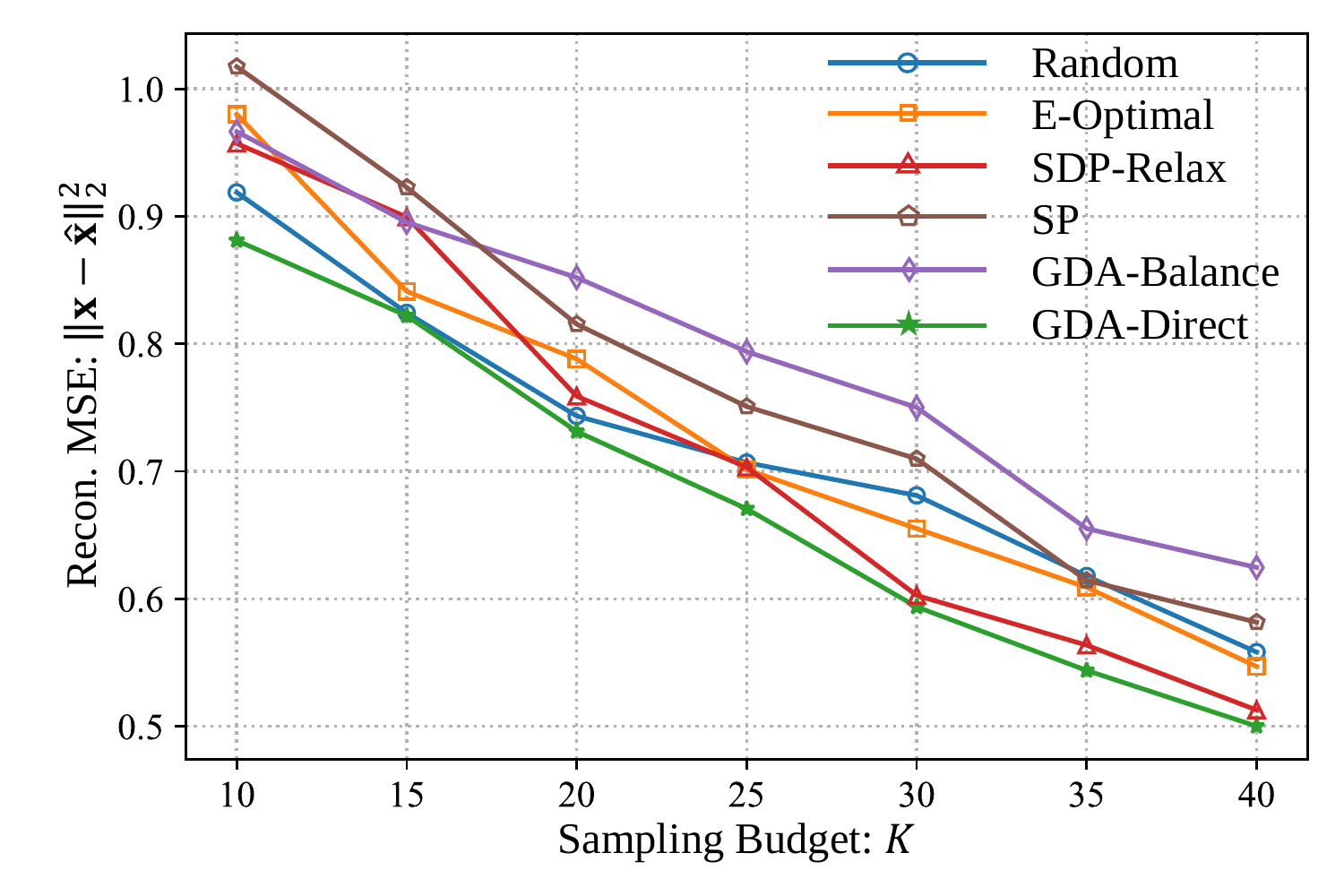}}
\subfigure{\includegraphics[width=0.31\textwidth]{mse-200-random-p010-lowpass01.pdf}}
\subfigure{\includegraphics[width=0.31\textwidth]{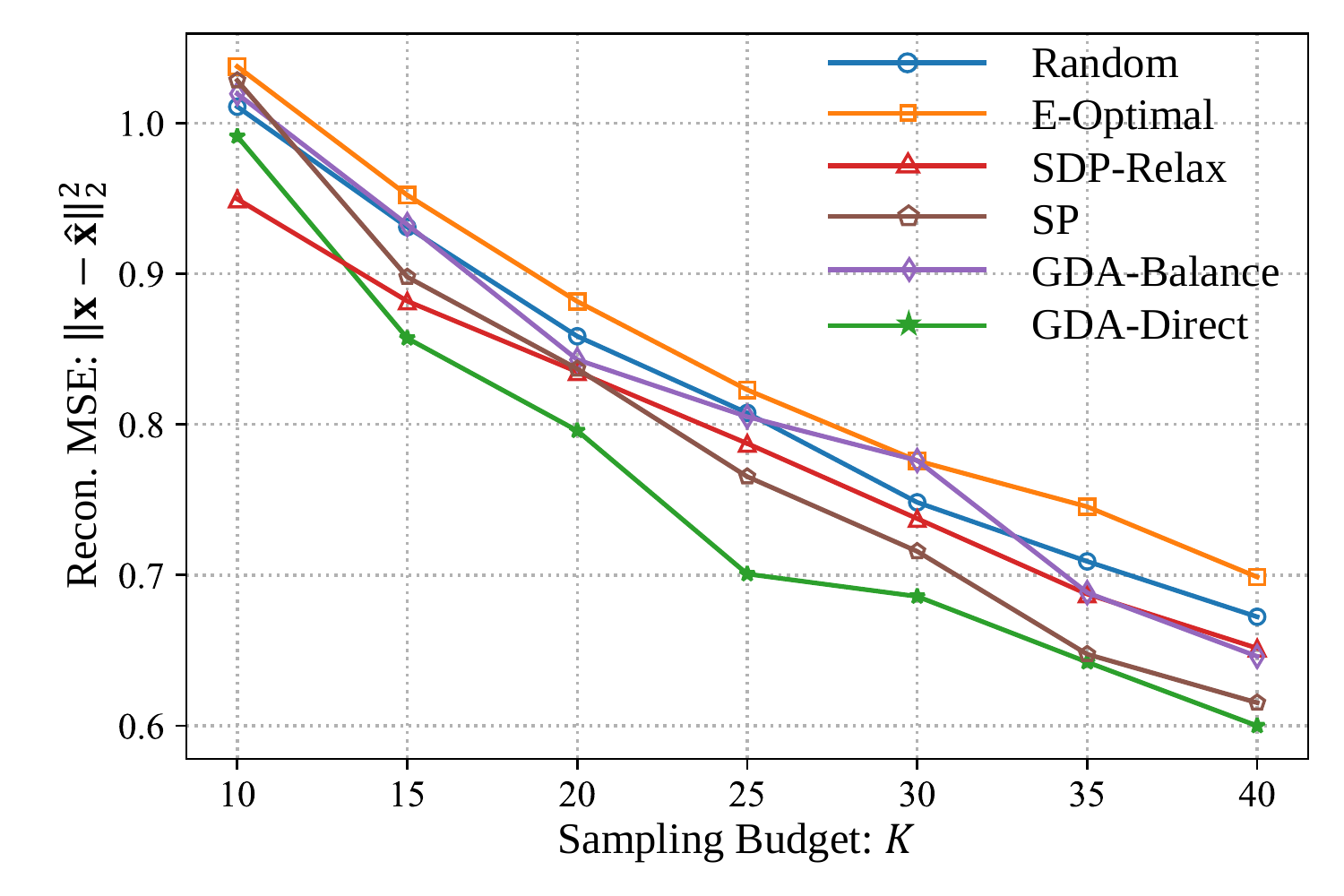}}
\caption{Reconstruction MSE of GS1 on random graph with (left) $p=0.05$, (mid) $p=0.1$, and (right) $p=0.15$.}
\label{fig:mse-results-sparse}
\end{figure*}
% ----- to be removed -----

\begin{figure}[t]
\centering
\includegraphics[width=0.45\textwidth]{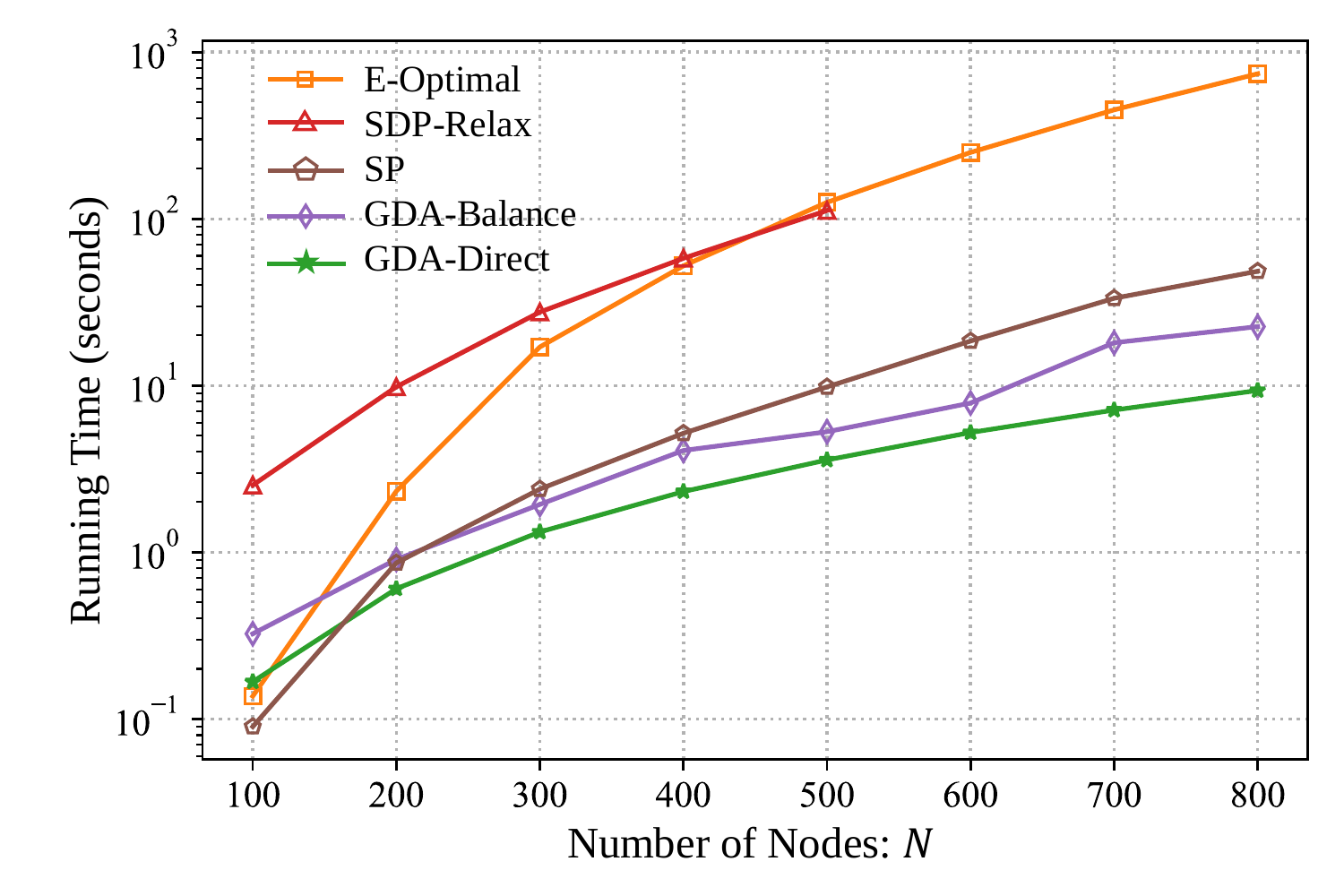}
\caption{Running time on random graphs with different sizes.}
\label{fig:run-time}
\end{figure}

% \begin{figure}[t]
% \centering
% \includegraphics[width=0.45\textwidth]{example-image-a}
% \caption{Reconstruction MSE of diffusion signals with different $T$.}
% \label{fig:diffusion-mse}
% \end{figure}
We tested the proposed sampling algorithm on synthetic directed graphs.
Our experimental platform was Ubuntu 18.04 server with a 32-core AMD Ryzen 3970X CPU and 250 GB memory. 
All algorithms were
implemented with Python 3.8.

\vspace{0.05in}
\noindent
\textbf{Graph Structure}. We randomly generated Erd{\"o}s R{\'e}nyi random graphs with $N=200$ nodes for experiments.
For any ordered pair of nodes $u$ and $v$, a directed edge $( u,v )$ was generated with probability $p$.
% \red{use past tense when describing experimental procedures.}
% The second is the directed growing graph in \cite{krapivsky2001organization}, where the graph is built by adding one node at a time.
% The newly added node connects to one of the existing nodes with the probability proportional to that node's total degree (in-degree $+$ out-degree).
To satisfy the assumption in Sec.\;\ref{sec:prelim} that there is at least one node that can be reached by any other nodes, we first generated a random graph with $N-1$ nodes.
Then, we manually added the last node $v_N$ and directed edges $(v_i, v_N)$ for $i=1,...,N-1$.
Further, we randomly chose a node $v_i$ ($i\neq N$), and added a directed edge $(v_N, v_i)$, such that there were no sink nodes with zero out-degrees in the graph.
We independently generated the weight of each edge from uniform distribution in $[0,1]$ and then normalized the weights such that $\bar{\W}\1 = \1$.

\vspace{0.05in}
\noindent
\textbf{Graph Signal}. 
We considered three types of graph signals.
\begin{enumerate}
\item (GS1) We used the eigen-decomposition $\L_{rw}^\top \L_{rw} = \U \tilde{\bLambda} \U^\top$ where the eigenvectors and eigenvalues are $\u_i$'s and $\tilde{\lambda}_i$'s, and generated a random bandlimited graph signal $\x = \sum_{i=1}^{m} c_i \u_i$, where $\tilde{\lambda}_m$ was the cutoff frequency, and each $c_i$ was independently generated from normal distribution $\cN(0, 1)$. %\red{how do u choose $c_i$'s?}
Here, $m=\lceil 0.1 N \rceil$.
    % We randomly generate $c_1,\cdots,c_m$.
    % and unified $c_i$ with $c_i / (\sum_{i=1}^{m}\Vert c_i \Vert_2)$ such that the absolute value of $\x$ is 1.
    % Then, we add a white noise $\n$ with the covariance matrix $\sigma^2\I$ to $\x$, and obtain the graph signal $\x_1 = \x + \n$.
\item (GS2) We generated graph signals following a normal distribution $\cN(\0, (\L_{rw}^\top \L_{rw} + \omega \I)^{-1})$ where $\omega = 0.1$.
    % Then, we normalize the signal with $\x \leftarrow \frac{\x - \text{mean}(\x)}{\text{std}(\x)}$ such that the true signal $\x$ has unit power.
\item (GS3) We generated graph signals $\x$ through a diffusion process.
Specifically, we randomly generated an initial signal $\x(0)$ from the normal distribution $\cN(0, \I)$.
Then, we followed a diffusion process $\x(t) = (1-\alpha) \x(t-1) + \alpha\bar{\W} \x(t-1)$ for $T$ steps, where $\alpha$ is a parameter, and the graph signal is defined as $\x = \x(T)$.
The intuition here is that, as $T\rightarrow\infty$, all entries of $\x(T)$ converge to the same value \cite{degroot1974reaching}. 
Thus, the smoothness prior $S(\x(T))$ in \eqref{eqn:gsv-digraph} approaches $0$. 
\end{enumerate}
For each generated signal, we normalized via $\x \leftarrow \frac{\x - \text{mean}(\x)}{\sqrt{N}\cdot\text{std}(\x)}$, such that $\Vert \x \Vert_2^2 = 1$.

\vspace{0.05in}
\noindent
\textbf{Baseline Methods}. 
We compared the proposed sampling method \texttt{GDA-Direct} with the following baseline methods.
% >>>>> to be replaced >>>>>
\begin{itemize}
\item \texttt{Random}: This method randomly selects sampling nodes.
\item \texttt{E-optimal} \cite{chen2015discrete}: Using the E-optimality criterion, this method greedily selects sampling nodes one-by-one.
\item \texttt{SDP-Relax} \cite{vandenberghe1996semidefinite}: This method relaxes the original integer constraint $A_{i,i} \in \{0,1\}$ in \eqref{eqn:ori-prob} to a continuous constraint $0 \le A_{i,i} \le 1$. Consequently, the relaxed problem can be formulated as a semi-definite programming (SDP) problem \cite{vandenberghe1996semidefinite}
\begin{align}
\min_{A_{i,i}} & ~ -l 
\label{eqn:sdp} \\
\mbox{s.t.} & ~\sum_{i} A_{i, i} \E_{i} + \L_{rw}^\top \L_{rw} - l \cdot \I \succeq \0,
\nonumber \\
& ~ 0 \le A_{i,i} \le 1,  \forall i, ~~~~~~~ \sum_i A_{i, i} = K.
\nonumber
\end{align}
Matrix $\E_{i}$ has only non-zero entry $E_{i,i}=1$.
Since the problem in (\ref{eqn:sdp}) is convex, we used \texttt{cvxopt} to obtain a solution. 
Given an optimal solution $A_{i,i}^*$ to \eqref{eqn:sdp}, we selected the $K$ largest $A^*_{i, i}$ and set them to $1$ and other $A^*_{i, i}$ to $0$.
\item \texttt{GDA-Balance} \cite{dinesh2020}: As discussed in Sec.\;\ref{sec:gda-alg}, this method treats $\L_{rw}^\top \L_{rw}$ as a generalized Laplacian for a signed graph, and uses the algorithm in \cite{dinesh2020} for sampling.
\item \texttt{SP} \cite{anis2016efficient}: This method greedily selects the sampling nodes with the graph spectral proxy.
We use $\L_{rw}^T \L_{rw}$ as the variation operator in \cite{anis2016efficient}.
\end{itemize}
% <<<<< to be replaced <<<<<
% \texttt{Random} (randomly select samples), \texttt{E-optimal} \cite{chen2015discrete}, \texttt{SDP-Relax} \cite{vandenberghe1996semidefinite}, \texttt{GDA-Balance} \cite{dinesh2020}, and \texttt{SP} \cite{anis2016efficient}.
% >>>>> to be replaced >>>>>
In the experiments, we randomly generated 5 random graphs %\red{I thought u said earlier that u generated an E-R random graph witn 200 nodes. this description does not look consistent.}
and 3000 graph signals over each graph.
% For each generated graph signal $\x$, we add a white noise $\n$ with the covariance matrix $\sigma^2\I$ to $\x$.
% The noise parameter $\sigma^2$ is set to $0.05$ during the experiments.
$T$ and $\alpha$ for GS3 were set to $50$ and $0.1$ unless otherwise specified.
We adopted the signal reconstruction scheme (\ref{eqn:recon}) in Sec.\;\ref{sec:sig-recon}, and set the hyper-parameter $\mu=0.001$.
Next, we show the average results over 15000 simulation runs.

We first compare the reconstruction MSE $\Vert \x - \hat{\x}\Vert_2^2$ of different sampling methods.
The results with different graph signals are shown in Fig.\;\ref{fig:mse-results-N200}.
For all three types of graph signal and sampling budget $K$, the proposed \texttt{GDA-Direct} performed better than other baseline sampling methods.
In particular, when the sampling budget $K$ is small, the superiority of \texttt{GDA-Direct} over other baseline methods is more obvious.
For example, \texttt{GDA-Direct} decreases the reconstruction MSE by 8.6\% for low-pass signal (GS1) with $K=20$ samples, and that by 11.9\% for diffusion signal (GS3) with only $K=10$ samples, comparing to other baseline methods that do not use GDAS algorithm.

% ----- to be removed -----
We also compare \texttt{GDA-Direct} with other baseline methods on random graphs with different sparsity.
We adjustes $p$ from $0.05$ to $0.15$, and The resulting reconstruction errors for GS1 are shown in Fig.\;\ref{fig:mse-results-sparse}.
Similar results for other types of graph signals were observed and thus are omitted ere.
In Fig.\;\ref{fig:mse-results-sparse}, we see that \texttt{GDA-direct} outperformed other methods on all graphs with different sparsity.
These observations validate the effectiveness of \texttt{GDA-Direct}.
% ----- to be removed -----

Next, we compare the efficiency of the proposed \texttt{GDA-Direct} to other sampling methods.
In Fig.\;\ref{fig:run-time}, we plotted running time of different sampling methods on random graphs ($p=0.1$) with different sizes $N$.
The sampling budget was set to $K=0.3N$.
Since \texttt{Random} method randomly samples from the graph with negligible  computational cost, we omitted its running time for clarity.
% We can see that when the graph is large ($N \ge 300$), the proposed \texttt{GDA-Direct} runs much faster than other baseline methods.
When the graph is large ($N \ge 300$), we see that both \texttt{GDA-Direct} and \texttt{GDA-Balance} run faster than other methods due to the linear time complexity of the GDA algorithm \cite{bai2020fast}.
Further, \texttt{GDA-Direct} ran faster than \texttt{GDA-Balance}; specifically, when the graph size is  $800$, \texttt{GDA-Direct} is 1.4 times faster than \texttt{GDA-Balance}.
This is because the pre-computation of $\delta$ and $\rho$ in \texttt{GDA-Direct} is more efficient than the balancing procedure in \texttt{GDA-Balance}.
Note that when the graph size was small ($N \le 100$), \texttt{E-Optimal} and \texttt{SP} were faster than the proposed \texttt{GDA-Direct}.
This is because \texttt{E-Optimal} and \texttt{SP} depend heavily on matrix operations, \eg, SVD and LOBPCG.
In Python, these matrix operations are optimized and parallelized, and thus, \texttt{E-Optimal} and \texttt{SP} have advantages for smaller graphs.
These observations further validate the efficiency and effectiveness of the proposed sampling algorithm on large graphs ($N \ge 300$).

%% file: conclude.tex
In this paper, we study the sampling problem on directed graphs.
We propose a graph signal reconstruction scheme using graph shift variation as the smoothness regularizer, and formulate a sampling problem under the E-optimality criterion.
For the formulated sampling problem, we propose a fast algorithm based on the Gershgorin disc alignment algorithm, which does not require eigendecomposition.
Experiment results on synthetic graphs show that the proposed sampling algorithm decreases the reconstruction MSE by at least 8.6\% and speeds up about 1.4 times compared to other baseline methods.

%% file: appendix_proof_prop1.tex
To prove Proposition \ref{prop:main-lb}, we first have the following lemmas
\begin{lemma}
\label{lemma:complex-eig-le}
For a real matrix (not necessarily symmetric) $\M=\P\bLambda\P^{-1}$ and a non-negative invertible diagonal matrix $\S$ such that $\lambda^-_{\min}(\S\M\S^{-1}) \ge 0$ 
% (\red{smallest Gershgorin disc left-end is often not $\geq 0$ even if $\M$ is PSD.}
% \blue{This is True for a general matrix $\M$ and $\S$.
% But recall the GDA algorithm. Initially, the discs of $\M$ (indeed $\delta\A + \rho\L_{rw}$) are aligned at 0.
% During the alignment process, we shift one disc at a time, and scale this disc's left end at threshold $T$.
% Then, we adjust all other discs' radius so that they are aligned at $T$ as well. This process repeats until we shift $K$ discs.
% From this process, we can see that we are always ensuring the discs' left ends of $\S \M \S^{-1}$ being non-negative.
% We can exploit this advantage of the GDA algorithm, and make this assumption.})
, we have
\begin{equation}
\frac{1}{\gamma^2(\M)}(\lambda^-_{\min}(\S\M\S^{-1}))^2 \le \lambda_{\min} (\M^\top \M),
\end{equation}
where $\gamma(\M) \triangleq \frac{\lambda_{\max}(\P^\top\P)}{\lambda_{\min}(\P^\top \P)} > 0$ is the condition number of matrix $\P$ and it is a function of matrix $\M$.
\end{lemma}
\begin{proof}
    let $\lambda^a_{\min}(\M) = a + b i \in \mathbb{C}$ be the eigenvalue of matrix $\M$ with the smallest absolute value, \ie, $0 \le \vert \lambda^a_{\min} (\M) \vert \le \vert \lambda(\M) \vert$. 
    % \red{this notation $\lambda_{\min}$ to mean the smallest eigenvalue in magnitude is confusing, since the the smallest absolute value eigenvalue can have a rather large real component (and a small imaginary component). u can write $\lambda_{\min}^a$ instead to avoid the confusion, for example.}
    Here, $i$ is the imaginary unit. 
    Since $\S\M\S^{-1}$ has the same eigenvalues as $\M$, $\lambda^a_{\min}(\M)$ is also the smallest eigenvalue of $\S\M\S^{-1}$ in terms of absolute value. 
    According to the Gershgorin disc theorem, eigenvalue $\lambda^a_{\min}(\M)$ must lie in at least one disc of $\S\M\S^{-1}$. 
    Since we have assume that $\lambda^-_{\min}(\S\M\S^{-1}) \ge 0$, we have $a = \text{Re}(\lambda_{\min}(\M)) \ge \lambda^-_{\min}(\S\M\S^{-1}) \ge 0$. Consequently, we have
    $$
    \lambda^-_{\min}(\S\M\S^{-1}) \le a \le \vert \lambda^a_{\min}(\M) \vert.
    $$
    With the eigendecomposition $\M = \P \bLambda \P^{-1}$, we define the condition number of $\P$ to be a function of matrix $\M$, \ie, $\gamma(\M) \triangleq \frac{\lambda_{\max}(\P^\top \P)}{\lambda_{\min}(\P^\top \P)} > 0$.
    According to Theorem 1 in \cite{ruhe1975closeness}, we have
    $$
    |\lambda^a_{\min}(\M)| \le \gamma(\M) \sigma_{\min}(\M) = \gamma(\M)\sqrt{\lambda_{\min}(\M^\top\M)},
    $$
    where $\sigma_{\min}$ is the smallest singular value of $\M$.
    Combining the above two inequality, we have
    $$
    0 \le \lambda^-_{\min}(\S\M\S^{-1}) \le \gamma(\M)\sqrt{\lambda_{\min}(\M^\top\M)}, \text{ and}
    $$
    $$
    \Rightarrow~ \frac{1}{\gamma^2(\M)}(\lambda^-_{\min}(\S\M\S^{-1}))^2 \le \lambda_{\min} (\M^\top \M).
    $$
    This ends the proof.
    % Furthermore, the eigenvalues of $\M$ and those of $\M^\top \M$ satisfy that $\vert \lambda(\M) \vert^2 = \lambda (\M^\top \M) \in \mathbb{R}$. 
    % \red{This is not true if $\M$ is not symmetric. As a counter-example, try $\M = [1 ~2; ~3 ~4]$.}
    % Thus, we have $\lambda_{\min}(\M^\top\M) = \vert \lambda_{\min}(\M) \vert^2$. With the above relation that $(\lambda^-_{\min}(\S\M\S^{-1}))^2 \le \vert \lambda_{\min} (\M) \vert^2$, we have $(\lambda^-_{\min}(\S\M\S^{-1}))^2 \le \lambda_{\min}(\M^\top\M)$.
\end{proof}
\begin{lemma}
\label{lemma:suf-sml-eig}
For two symmetric matrices $\M_1, \M_2\in \mathbb{R}^{N\times N}$, if $ \M_1 - \M_2 \succeq 0$, we have $\lambda_{\min}(\M_1) \ge \lambda_{\min}(\M_2)$.
\end{lemma}
% \red{This is well known, and is typically written as $\M_1 - \M_2 \succeq 0$, or $\M_1 \succeq \M_2$.}
\begin{proof}
    We refer readers to Proposition 2 in \cite{dinesh2020} for the proof.
    % Let $\lambda_{\min}(\M_1)$ be the smallest eigenvalue of $\M_1$ and $\u_1$ be its corresponding eigenvector. Similarly, we also denote $\lambda_{\min}(\M_2)$ and $\u_2$ be the smallest eigenvalue of $\M_2$ and its corresponding eigenvector. With the Rayleigh quotient theorem, for all unit vector $\u$, we have
    % $$
    % \begin{aligned}
    %     \lambda_{\min}(\M_1) &= \u_1^\top \M_1 \u_1 \le \u^\top \M_1 \u,\text{ and}\\
    %     \lambda_{\min}(\M_2) &= \u_2^\top \M_2 \u_2 \le \u^\top \M_2 \u,
    % \end{aligned}
    % $$
    % respectively. Since we assumed that $\u^\top (\M_1 - \M_2) \u \ge 0$ for all unit vector $\u$, we have $\u_1^\top \M_2 \u_1 \le \u_1 \M_1 \u_1$. Consequently, we have
    % $$
    % \lambda_{\min}(\M_2) = \u_2^\top \M_2 \u_2 \le \u_1^\top \M_2 \u_1 \le \u_1^\top \M_1 \u_1 = \lambda_{\min}(\M_1).
    % $$
    % This ends the proof.
\end{proof}

\begin{lemma}
\label{lemma:ub-lambda-max}
For any diagonal matrix $\A$ such that $A_{i, i} \in \{0, 1\}$ and $\text{tr}(\A) = K$, we have
    $$
    \lambda_{\max} (\L_{rw}^\top \A + \A \L_{rw}) \le 3 + \max_{i} \sum_{j=1}^{K-1} \bar{W}_{[j], i},
    $$
    where $W_{[j], i}$ is the $j$-th largest element in the $i$-th column of $\bar{\W}$.
\end{lemma}
\begin{proof}
    Let $\B = (\L_{rw}^\top \A + \A \L_{rw})$ and $\cT = \{i | A_{i, i} = 1\}$.
    Note that $(\L_{rw}^\top \A)^\top = \A \L_{rw}$ and
    $$
    (\A\L_{rw})_{i, j} = (\L_{rw}^\top \A)_{j, i} =
    \left\{
        \begin{array}{lc}
            (L_{rw})_{i, j} & \text{if } i \in \cT,\\
            0       & \text{otherwise}.
        \end{array}
    \right.
    $$
    $\B$ is a symmetric matrix, and $\lambda_{\max}(\B)$ is a real value. With the entry of $\A\L_{rw}$ and $\L_{rw}^\top \A$ above, we have
    $$
    B_{i, j} =
    \left\{
        \begin{array}{lc}
            (L_{rw})_{i, j} + (L_{rw})_{j, i} & \text{if } i, j \in \cT,\\
            (L_{rw})_{i, j}            & \text{if } i \in \cT \text{ and } j \notin \cT, \\
            (L_{rw})_{j, i}            & \text{if } j \in \cT \text{ and } i \notin \cT, \\
            0                   & i, j \notin \cT.
        \end{array}
    \right.
    $$
    To upper bound the largest eigenvalue $\lambda_{\max}(\B)$, we need to study the right ends of $\B$'s Gershgorin discs. For the $i$-th row of $\B$, if $i\in\cT$, we have $B_{i, i} = 2 (L_{rw})_{i, i} = 2$, and the center of the corresponding disc is $(2, 0)$. For the radius of the disc corresponding to the $i$-th row of $\B$ ($i\in \cT$), we have
    $$
    \begin{aligned}
        \sum_{j\neq i} \vert B_{i, j}\vert &= \sum_{j\in\cT, j\neq i} \vert (L_{rw})_{i, j} + (L_{rw})_{j, i} \vert + \sum_{j\notin\cT} \vert (L_{rw})_{i, j} \vert\\
        &=\sum_{j\in\cT, j\neq i} \vert -\bar{W}_{i, j} - \bar{W}_{j, i} \vert + \sum_{j\notin\cT} \vert - \bar{W}_{i, j} \vert\\
        &=\sum_{j\neq i} \bar{W}_{i, j} + \sum_{j\in\cT, j\neq{i}} \bar{W}_{j, i} = 1 + \sum_{j\in\cT, j\neq{i}} \bar{W}_{j, i}.\\
    \end{aligned}
    $$
    Thus, the right end of the $i$-th disc is $3 + \sum_{j\in\cT, j\neq{i}} \bar{W}_{j, i} \le 3 + \sum_{j=1}^{k-1} \bar{W}_{[j], i}$.
    Similarly, for the disc of the $i$-th row where $i\notin \cT$, the center of this disc is $(0, 0)$, and the radius is $\sum_{j\in\cT, j\neq{i}} \bar{W}_{j, i}$.
    Then, the right end of this disc is $\sum_{j\in\cT, j\neq{i}} \bar{W}_{j, i} \le \sum_{j=1}^{k-1} \bar{W}_{[j], i}$.
    From the Gershgorin circle theorem, the largest eigenvalue of $\B$ is smaller than the largest (upper bounds of) disc's right end. Therefore, we have
    $$
        \begin{aligned}
            \lambda_{\max} (\B) &\le \max_{i} \max \{ 3 + \sum_{j=1}^{k-1} \bar{W}_{[j], i}, \sum_{j=1}^{k-1} \bar{W}_{[j], i} \}\\
            &= 3 + \max_{i} \sum_{j=1}^{k-1} \bar{W}_{[j], i}.
        \end{aligned}
    $$
    This ends the proof.
\end{proof}

% ------- faster bound -------
% \begin{proof}
%     Let $\lambda_{\max}(\M)$ be the eigenvalue with the largest absolute value among all of $\M$'s eigenvalues. Note that $(\A \L)^\top = \L^\top \A^\top  = \L^\top \A$.
%     Thus, we have $\lambda_{\max}(\L^\top \A) = \lambda_{\max}((\A \L)^\top) = \lambda_{\max}(\A \L)$. Furthermore, $\lambda_{\max} (\L^\top \A + \A \L)$, $\vert \lambda_{\max}(\L^\top \A) \vert$ and $\vert \lambda_{\max}(\A \L) \vert$ are the spectral norm of $(\L^\top \A + \A \L)$, $\L^\top \A$, and $\A \L$, respectively. Then, we have
%     $$
%     \lambda_{\max} (\L^\top \A + \A \L) \le \vert\lambda_{\max}(\L^\top \A)\vert + \vert\lambda_{\max}(\A \L)\vert = 2  \vert\lambda_{\max}(\A \L)\vert.
%     $$

%     Note that the $i$-th row of $\A\L$ is the $i$-th row of $\L$ if $D_{i, i} = 1$, and $\0^\top$ otherwise. Therefore, $\A\L$ only has two Gershgorin discs. The first disc is centered at $(0,0)$ on the complex plane, with radius $0$. The second disc is centered at $(1, 0)$ with radius $1$. From the Gershgorin disc theorem, the eigenvalue $\lambda_{\max}$ must lies in one of the two discs. If it lies in the second disc, the value $\vert \lambda_{\max}(\A \L) \vert \le 2$, because the second disc goes through the point $(0, 0)$ and $2$ is the diameter of the disc. If it lies in the first disc, this means that $\vert \lambda_{\max}(\A \L) \vert = 0 < 2$. Thus, we have $\lambda_{\max} (\L^\top \A + \A \L) \le 4$.
% \end{proof}
% ------- faster bound --------
With the above three lemmas, we prove Proposition \ref{prop:main-lb} as follows.
\begin{proof}
    % Since matrix $\A$ is a diagonal matrix and $A_{i, i} = \{ 0, 1 \} \ge 0$, for the $i$-th disc of matrix $(\delta\A + \rho \L)$, its center is $( \delta A_{i, i} + \rho, 0)$ and its radius is $\rho \sum_{j\neq i} \vert L_{i,j} \vert = \rho \sum_{j\neq i} \vert W_{i,j} \vert = \rho $. From (\ref{eqn:param-delta-rho}), we have $\delta,\rho \ge 0$. Therefore, the left end of the $i$-th disc of matrix $(\delta\A + \rho \L)$ is $\delta A_{i, i} + \rho - \rho = \delta A_{i, i} \ge 0 $. When performing the GDA algorithm on $\delta\A + \rho \L$, we can always have $\lambda^-_{\min}(\S(\delta\A + \rho \L)\S^{-1}) \ge 0$. Therefore, according to Lemma \ref{lemma:complex-eig-le}, we have
    Consider the eigendecomposition $(\delta \A + \rho\L_{rw}) = \P \bLambda \P^{-1}$. As $\delta$ and $\rho$ are given constant, this eigendecomposition only depends on $\A$, and we denote $\gamma(\delta \A + \rho\L_{rw}) = \gamma(\A) = \frac{\lambda_{\max}(\P^\top \P)}{\lambda_{\min}(\P^\top \P)}$.
    As we assume that $\lambda^-_{\min}(\S(\delta \A + \rho\L_{rw})\S^{-1}) \ge 0$, according to Lemma \ref{lemma:complex-eig-le}, we have
    $$
    \begin{aligned}
        &\frac{(\lambda^-_{\min}(\S(\delta\A + \rho \L_{rw})\S^{-1}))^2}{\gamma^2(A)} \\
        &\le \lambda_{\min} \big((\delta\A + \rho \L_{rw})^\top (\delta\A + \rho \L_{rw})\big).
    \end{aligned}
    $$
    Furthermore, we denote $\gamma_{\max} = \max_{\A} \gamma(\A)$ over all possible diagonal matrix $\A$, such that $A_{i, i} = \{ 0, 1 \}$ and $\text{tr}(\A)= K$. Therefore, we have
    \begin{align}
        &\frac{(\lambda^-_{\min}(\S(\delta\A + \rho \L_{rw})\S^{-1}))^2}{\gamma^2_{\max}} \le \frac{(\lambda^-_{\min}(\S(\delta\A + \rho \L_{rw})\S^{-1}))^2}{\gamma^2(A)} \nonumber\\
        &\le \lambda_{\min} \big((\delta\A + \rho \L_{rw})^\top (\delta\A + \rho \L_{rw})\big).\label{eqn:prof-prop-half}
    \end{align}

    Next, we aim to show that $\lambda_{\min} (\A + \mu \L_{rw}^\top \L_{rw}) \ge  \lambda_{\min} \big((\delta\A + \rho \L_{rw})^\top (\delta\A + \rho \L_{rw})\big)$. To achieve this, with Lemma \ref{lemma:suf-sml-eig}, we aim to show that for any unit vector $\u$, the quadratic form
    \begin{equation}
        \label{eqn:quad-form}
        f = \u^\top [(\A + \mu \L_{rw}^\top \L_{rw}) - (\delta\A + \rho \L_{rw})^\top (\delta\A + \rho \L_{rw})] \u \ge 0.
    \end{equation}

    Note that for matrix $\A$, we have $\A = \A^\top$, and $\A$ is an idempotent matrix, \ie, $\A^\top \A = \A\A = \A$. Consequently, we have
    $$
    (\delta\A + \rho \L_{rw})^\top (\delta\A + \rho \L_{rw}) = \delta^2 \A + \rho^2 \L_{rw}^\top \L_{rw} + \delta\rho(\L_{rw}^\top \A + \A\L_{rw}).
    $$
    Plugging the above equation into the quadratic form (\ref{eqn:quad-form}), we have
    \begin{align}
        &f=\u^\top [(\A + \mu \L_{rw}^\top \L_{rw}) - (\delta\A + \rho \L_{rw})^\top (\delta\A + \rho \L_{rw})] \u \nonumber\\
        = &\underbrace{(1 - \delta^2) \u^\top \A \u + (\mu - \rho^2) \u^\top \L_{rw}^\top \L_{rw} \u}_{f_1} \nonumber\\
        &- \underbrace{\delta\rho \u^\top (\L_{rw}^\top \A + \A\L_{rw}) \u}_{f_2}. \label{eqn:f}
        % =& \epsilon\mu \cdot \u^\top \A \u + (\mu - \rho^2) \cdot \u^\top \L^\top \L \u - \delta\rho \cdot \u^\top (\L^\top \A + \A\L) \u \\
        % =& \epsilon \rho^2 \cdot \u^\top \A \u + (\mu - \rho^2) \cdot \u^\top (\L^\top \L + \epsilon \A) \u - \delta\rho \cdot \u^\top (\L^\top \A + \A\L) \u
    \end{align}
    When $\mu \le 1$, if $\delta = \sqrt{1 - \epsilon \cdot \mu}$, the first two terms in (\ref{eqn:f}), \ie, $f_1$, becomes
    $$
    \begin{aligned}
        f_1 &= (1 - \delta^2) \u^\top \A \u + (\mu - \rho^2) \u^\top \L_{rw}^\top \L_{rw} \u \\
        &= \epsilon\mu \u^\top \A \u + (\mu - \rho^2) \u^\top \L_{rw}^\top \L_{rw} \u \\
        &= \epsilon \rho^2 \cdot \u^\top \A \u + (\mu - \rho^2) \cdot \u^\top (\L_{rw}^\top \L_{rw} + \epsilon \A) \u \\
        &\ge (\mu - \rho^2) \cdot \u^\top (\L_{rw}^\top \L_{rw} + \epsilon \A) \u.
    \end{aligned}
    $$
    With the Rayleigh quotient theorem, we have $\u^\top (\L_{rw}^\top \L_{rw} + \epsilon \A) \u \ge \lambda_{\min} (\L_{rw}^\top \L_{rw} + \epsilon \A)$. Since $\A$ is the optimization variable, we have
    \begin{equation}
        \label{eqn:lb-f1}
        f_1 \ge (\mu - \rho^2) \cdot \big(\min_{\A} \lambda_{\min} (\L_{rw}^\top \L_{rw} + \epsilon \A)\big).
    \end{equation}
    For $f_2$ in (\ref{eqn:f}), according to the Rayleigh quotient theorem and lemma \ref{lemma:ub-lambda-max}, we have
    \begin{align}
        f_2 &= \delta\rho \u^\top (\L_{rw}^\top \A + \A\L_{rw}) \u \le \delta\rho \lambda_{\max}(\L_{rw}^\top \A + \A\L_{rw}) \nonumber\\
        &\le \delta\rho (3 + \max_{i} \sum_{j=1}^{K-1} \bar{W}_{[j], i}).\label{eqn:ub-f2}
    \end{align}
    Plugging (\ref{eqn:lb-f1}) and (\ref{eqn:ub-f2}) into (\ref{eqn:f}), we have
    $$
    f \ge (\mu - \rho^2) \cdot \big(\min_{\A} \lambda_{\min} (\L_{rw}^\top \L_{rw} + \epsilon \A)\big) - \delta\rho (3 + \max_{i} \sum_{j=1}^{k-1} \bar{W}_{[j], i}).
    $$
    When $\rho = \frac{-\delta c + \sqrt{\delta^2 c^2 + 4\mu}}{2}$, the right hand side of the above inequality is zero. Consequently, we have $f \ge 0$ when $\mu \le 1$.

    When $\mu > 1$, if $\rho = \sqrt{\mu - \epsilon}$, we have
    $$
    \begin{aligned}
        f_1 &= (1 - \delta^2) \u^\top \A \u + (\mu - \rho^2) \u^\top \L_{rw}^\top \L_{rw} \u \\
        &= (1 - \delta^2) \u^\top \A \u + \epsilon \cdot \u^\top \L_{rw}^\top \L_{rw} \u \\
        &= \epsilon \delta^2 \cdot \u^\top \L_{rw}^\top \L_{rw} \u + (1 - \delta^2) \u^\top (\A + \epsilon \L_{rw}^{\top} \L_{rw}) \u\\
        &\ge (1 - \delta^2) \u^\top (\A + \epsilon \L_{rw}^{\top} \L_{rw}) \u \\
        &\ge (1 - \delta^2) \cdot \lambda_{\min} (\A + \epsilon \L_{rw}^{\top} \L_{rw})\\
        &\ge (1 - \delta^2) \cdot \big(\min_{\A} \lambda_{\min} (\A + \epsilon \L_{rw}^{\top} \L_{rw}) \big).
    \end{aligned}
    $$
    Combining the upper bound of $f_2$ in (\ref{eqn:ub-f2}), we have
    $$
    f \ge (1 - \delta^2) \cdot \big(\min_{\A} \lambda_{\min} (\A + \epsilon \L_{rw}^{\top} \L_{rw}) \big) - \delta\rho (3 + \max_{i} \sum_{j=1}^{k-1} \bar{W}_{[j], i}).
    $$
    When $\delta = \frac{-\rho c + \sqrt{\rho^2 c^2 + 4}}{2}$, the right hand side of the above inequality is zero. Consequently, we have $f \ge 0$ when $\mu > 1$.

    In summary, with $\delta$ and $\rho$ in (\ref{eqn:param-delta-rho}) and (\ref{eqn:param-delta-rho-mu-ge1}), we can always ensure that $f \ge 0$.
    According to Lemma \ref{lemma:suf-sml-eig}, we have $\lambda_{\min} (\A + \mu \L_{rw}^\top \L_{rw}) \ge  \lambda_{\min} \big((\delta\A + \rho \L_{rw})^\top (\delta\A + \rho \L_{rw})\big)$. Consequently, combining (\ref{eqn:prof-prop-half}), we have 
    % $\big(\lambda^-_{\min}(\S(\delta\A + \rho \L_{rw})\S^{-1})\big)^2 \le \lambda_{\min} (\A + \mu \L_{rw}^\top \L_{rw})$
    $$
    \frac{(\lambda^-_{\min}(\S(\delta\A + \rho \L_{rw})\S^{-1}))^2}{\gamma^2_{\max}} \le \lambda_{\min} (\A + \mu \L_{rw}^\top \L_{rw})
    $$
    This ends the proof.
\end{proof}